\RequirePackage{fix-cm}

\documentclass[twocolumn]{svjour3}       
\smartqed  
\usepackage{graphicx}

\usepackage{color}
\usepackage{greekbf}
\usepackage{amsmath,amssymb}
\usepackage{mathtools}

\usepackage{url,hyperref,lineno,microtype,subcaption}
\usepackage[onehalfspacing]{setspace}
\usepackage{graphicx}
\usepackage{epsfig}
\usepackage{amssymb}
\usepackage{mathptmx}      
\usepackage{Tabbing}
\usepackage{color}

\renewcommand{\rho}{\varrho}
\renewcommand{\phi}{\varphi}
\newcommand{\hp}{\widehat{p}}
\newcommand{\hq}{\widehat{q}}
\newcommand{\bbi}{\mathbf{i}}
\newcommand{\bbj}{\mathbf{j}}
\newcommand{\bbk}{\mathbf{k}}
\newcommand{\ovq}{\overline{q}}
\newcommand{\vL}{\varLambda}
\newcommand{\nn}{\nonumber}

\newcommand{\pat}{\partial_t}
\newcommand{\pax}{\partial_x}

\newcommand{\pal}{\partial_l}
\newcommand{\pam}{\partial_m}

\newcommand{\ve}{\varepsilon}
\newcommand{\Wst}{W_{\mathrm{st}}}
\newcommand{\WWc}{W_{\mathrm{c}}}
\newcommand{\Fe}{F_\mathrm{e}}
\newcommand{\Fp}{F_\mathrm{p}}
\newcommand{\fext}{f_\mathrm{ext}}
\newcommand{\Mext}{M_\mathrm{ext}}
\newcommand{\Ue}{U_\mathrm{\!e}}
\renewcommand{\Re}{R_\mathrm{\!e}}
\newcommand{\wRe}{\widehat{R_\mathrm{\!e}}}
\newcommand{\Ke}{K_\mathrm{\!e}}
\newcommand{\Hp}{\mathbb{H}_\mathrm{p}}
\newcommand{\cE}{{\mathcal E}}
\newcommand{\ddx}{\,\mathrm{d}x}
\newcommand{\dS}{\,\mathrm{d}S}
\newcommand{\skw}{\operatorname{skw}}
\newcommand{\Id}{\mathbb{I}}
\newcommand{\io}{\int\limits_\Omega}
\newcommand{\N}{\mathbb{N}}

\newcommand{\R}{\mathbb{R}}
\renewcommand{\H}{\mathbb{H}}

\DeclareMathAlphabet{\mathbf}{OT1}{cmr}{bx}{it}
\DeclareMathAlphabet{\bfit}{OT1}{cmr}{bx}{it}

\newcommand{\bydef}{\,\raise.050ex\hbox{\rm:}\kern-.025em\hbox{\rm=}\,}
\newcommand{\defby}{=\raise.075ex\hbox{\kern-.325em\hbox{\rm:}}\,}

\def\qed{\relax\ifmmode\hskip2em \Box\else\unskip\nobreak\hskip1em $\Box$\fi}
\newcommand {\SO} {\mathcal{SO}}
\newcommand {\GL} {\mathcal{G}l}

\newcommand {\tr}[1]{\mbox{tr}\, #1}

\newcommand{\sym}{\mathop{\mathrm{sym}}}

\journalname{}
\begin{document}
\title{Mathematical modeling of a Cosserat method in finite-strain holonomic
plasticity}


\titlerunning{Mathematical modeling of a Cosserat method in finite-strain holonomic
plasticity}        

\author{Thomas Blesgen         \and
        Ada Amendola$^{*}$ ($^{*}$Corresponding Author, Orcid ID: 0000-0002-2562-881X) 
}

\authorrunning{T. Blesgen, A. Amendola} 

\institute{Thomas Blesgen \at
             Department of Applied Mathematics, Bingen University of Applied Sciences, Bingen, Germany\\
           \and
           Ada Amendola\at
              Department of Civil Engineering, University of Salerno, Fisciano, Italy
\email{adaamendola1@unisa.it} 
}

%

\date{Submitted to Springer}

\maketitle

\begin{abstract}
This article deals with the mathematical derivation and the validation over
benchmark examples of a numerical method for the solution of a finite-strain
holonomic (rate-independent) Cosserat plasticity problem for
materials, possibly with microstructure.
Two improvements are made in contrast to earlier approaches: First, the
micro-rotations are parameterized with the help of an Euler-Rodrigues
formula related to quaternions. Secondly, as main result, a novel two-pass
preconditioning scheme for searching the
energy-minimizing solutions based on the
limited memory Broyden-Fletcher-Goldstein-Shanno quasi-Newton method is
proposed that consists of a predictor step and a corrector-iteration.
After outlining the necessary adaptations to the model, numerical simulations
compare the performance and efficiency of the new and the old algorithm.
The proposed numerical model can be effectively employed for studying the
mechanical response of complicated materials featuring large size effects.

\end{abstract}


\section{Introduction}
\label{secintro}
In recent years, the scientific interest towards sophisticated
and heterogeneous materials featuring multiple internal length scales has grown
significantly, mainly due to the possibility of playing with
the internal microstructure of these materials to model and engineer structures
that exhibit properties not found in conventional materials (refer, e.g., to
\cite{Lakes91,bardellaSI} and references therein).
Such materials include cellular solids, fibrous and particle composites,
biological materials, robots,  and also building-scale systems made of masonry
structures
\cite{ZMP17,TACM14,TCPB18a,JAM10,JDP98,minga}.
The mechanical modeling of these materials and structures calls for the
introduction of degrees of freedom that are not accounted for in classical
continuum mechanics, typically  rotation of points (or micro-rotations)
and couple stresses \cite{BOOK95,Mindlin64,Eringen99}.
A viable continuum description of such phenomena is provided by the micropolar
theories of Cosserat continua \cite{Coss09}, which have been intensively
applied since their introduction in 1909 to a variety of different problems in
solid and structural mechanics, fluid dynamics, liquid crystals, granular
materials, powders, etc. (cf. \cite{Maugin10,Neff06,TJMPS18} for an overview).
Particularly interesting is the Cosserat modeling of chiral honeycomb lattices
with bending-dominated behavior whose mechanical response cannot be accurately
described by classical continuum theories due to large size effects,
\cite{ZMP17}. So far, physical models of these exciting materials have
been fabricated through additive manufacturing (AM) technologies in polymeric
materials and have been described through Cosserat elasticity,
\cite{ZMP17}.
The numerical model presented in this work allows for simulating the response
of ductile versions of such metamaterials, assuming radial loading and
holonomic plasticity, \cite{corradi1,corradi2,desimone}, which are, e.g.,
fabricated via AM techniques manual assembling methods employing metallic
materials, \cite{pentamode,theta1,theta2}.

Since the Cosserat model of a micropolar material induces sensitivity to the
microrotation strain gradient, such generalized continua are endowed
with an internal length scale such that localization zones have a finite width.
The Broyden-Fletcher-Goldfarb-Shanno (BFGS) algorithm is a well-known
quasi-Newton method where instead of storing the full Hessian matrix $H$
(a big matrix for large dimensions) an approximation is
computed by the sum of two rank-one matrices. In the limited-memory (L-BFGS)
variant, \cite{Nocedal80,LBFGS}, the approximation to $H$ is constructed from
a small number of vectors by a rank-one update formula, see Eqn.~(\ref{Hmul})
below. The resulting algorithm is still considered the state-of-the-art method
when huge systems of equations with a very large number of unknowns need to get
solved.

In \cite{Ble15}, a L-BFGS algorithm is developed for the solution of a
finite-strain rate-independent Cosserat model of finite plasticity.
Therein, the {\it elastic Cosserat micro-rotations} $\Re$ are parameterized by a
vector $\alpha=(\alpha_1,\alpha_2,\alpha_3)\in\R^3$ of Euler angles,
\begin{align}
\widetilde{\Re}(\alpha) &:= R_3(\alpha_3)R_2(\alpha_2)R_1(\alpha_1)\nn\\
\label{para2}
&:= \!\left(\begin{array}{c@{\;\;}c@{\;\;}c}
1             & 0             & 0\\
0             & \cos\alpha_3  & \sin\alpha_3\\
0             & -\sin\alpha_3 & \cos\alpha_3
\end{array}\right) \left(\begin{array}{c@{\;\;}c@{\;\;}c}
\cos\alpha_2  & 0             & -\sin\alpha_2\\
0             & 1             & 0\\
\sin\alpha_2  & 0             & \cos\alpha_2
\end{array}\right) 
\nn\\&
 \ \ \ \ \ \left(
\begin{array}{c@{\;\;}c@{\;\;}c}
\cos\alpha_1  & \sin\alpha_1  & 0\\
-\sin\alpha_1 & \cos\alpha_1  & 0\\
0             & 0             & 1\end{array}\right).
\end{align}


Two main criticisms of the approach in \cite{Ble15} are eminent. The first is
that Euler angles are not well-suited to parameterize the rotation group
$\SO(3)$ and have several shortcomings.
Especially the parameterization may degenerate and become non-unique.

In other areas of mechanics such as unmanned aerial vehicle (UAV) control,
quaternion-based descriptions have demonstrated their superior performance, see
\cite{AAMR13,DSM15}.
Therefore, in this article, the alternative parameterization 
\begin{equation}
\label{para1}
{\footnotesize \Re(q):=\left(\begin{array}{c@{\quad}c@{\quad}c}
q_0^2\!+\!q_1^2\!-\!q_2^2\!-\!q_3^2 & 2(q_1q_2\!-\!q_0q_3)
& 2(q_1q_3\!+\!q_0q_2)\\
2(q_1q_2\!+\!q_0q_3) & q_0^2\!-\!q_1^2\!+\!q_2^2\!-\!q_3^2
& 2(q_2q_3\!-\!q_0q_1)\\
2(q_1q_3\!-\!q_0q_2)& 2(q_2q_3\!+\!q_0q_1) & q_0^2\!-\!q_1^2\!-\!q_2^2\!+\!q_3^2
\end{array}\right)}
\end{equation}
is studied which is based on an Euler-Rodrigues vector $q=(q_1,q_2,q_3,q_4)$
defined on the unit sphere
\[ S^3:=\Big\{q=(q_1,q_2,q_3,q_4)\in\R^4\;\Big|\; |q|^2=1 \Big\}. \]

\noindent Formula~(\ref{para1}) goes back to historical work by
L. Euler in 1775, \cite{Eul1775}. The approach was independently
reinvented by Rodrigues in 1840, \cite{Rod1840}.
As was already discovered early,
it can also be derived from quaternion theory, \cite{Ham2}.

\noindent
The second major criticism to \cite{Ble15} is that the quasi-Newton iteration
may get stuck in a local minimum of the mechanical energy
without finding the global minimizer. Preconditioning of the
numerical scheme may help to speed up the code and correctly compute the
global minimizer. While there is vast literature on preconditioning in general,
only a few articles deal with preconditioning of the L-BFGS-method,
\cite{A06,EM12,DH18,ML13}, especially when directly related to energy
minimization, \cite{JBES04}.

\vspace*{2mm}
The first goal of this paper is to study the implications of
(\ref{para2}), (\ref{para1}) on the finite-strain Cosserat algorithm,
assuming radial loading and holonomic-type plasticity
\cite{corradi1,corradi2,desimone}.
Secondly, as main result, a two-step preconditioning strategy of the
L-BFGS algorithm is proposed that consists of a predictor step followed by a
corrector iteration for solving the time-discrete problem. This two-pass
approach effectively defines a non-linear preconditioning strategy.

This article is organized in the following way. In Section~\ref{secmodel},
the finite-strain Cosserat model is reviewed. Section~\ref{secER} derives
background theory on a quaternion-based Cosserat theory. Section~\ref{secprec}
revisits the L-BFGS update scheme and derives the aforementioned preconditioning
method. Section~\ref{secnum} performs various numerical tests, followed by a
discussion of the results and an outlook. At the end of the paper, a complete
list of symbols with explanations can be found. The generalization of the
present approach to more general cases of gradient-type plasticity
\cite{Gurtin02,CHM2002,bardellaJMPS,schatz,borokinni} 
is addressed to future work.

\section{The finite-strain Cosserat model of holonomic plastic materials with
microstructure} \label{secmodel}
The deformation mapping of the current material from the reference
configuration $\Omega\subset\R^3$ to the deformed state $\Omega_t$ is described
by a diffeomorphism $\phi\in\GL^+(3)$, for times $t\ge0$.
Throughout, $\Omega$ is assumed a smooth Lipschitz domain.

Assuming radial loading and holonomic-type plasticity \cite{corradi1,corradi2},
the fundamental relationship of the Cosserat approach is the multiplicative
decomposition
\begin{equation}
\label{split}
F=\Fe\Fp=\Re\Ue\Fp
\end{equation}
of the deformation tensor $F:=D\phi$, where $\Fe$, $\Fp$ are the elastic and the
plastic deformation tensors, $\Ue=\Re^tD\phi\Fp^{-1}\in\GL(3)$ is the
stretching component, and
\[ \Re\in\SO(3):=\{R\in\GL(3)\;|\;\det(R)=1,\;R^tR=\Id\} \]
are the micro-rotations. In (\ref{split}), $\Ue$ need not be symmetric and
positive definite, i.e. the decomposition $\Fe=\Re\Ue$ is in general {\it not}
the polar decomposition.

We fundamentally assume that the mechanical energy depends on the
elastic part $\Fe$ of the deformation, only.
With $\kappa$ denoting the density of the
(geometrically necessary) dislocations, it follows by frame indifference that
the stored mechanical energy is of the form, \cite{Kessel64},
\[ W(\Fe,\kappa)=\Wst(\Ue)+\WWc(\Ke)+V(\kappa), \]
where $\Ke=(\Re^t\partial_{x_l}\Re)_{1\le l\le3}$ is the (right) curvature
tensor, $\Wst$ denotes the stretching energy, $\WWc$ the curvature energy due to
bending and torsion of the material, and $V$ the energy of stored dislocations.
For these functionals we make the ansatz, cf. \cite{Neff06,Ble13},
\begin{align}
\label{Wstdef}
\Wst(\Ue) \,:=&\, \mu\|\sym\Ue-\Id\|^2+\mu_c\|\skw(\Ue-\Id)\|^2\nn\\
& +\frac\lambda2|\tr(\Ue-\Id)|^2,\\
\WWc(q) \,:=&\, \mu_2\|\Ke(q)\|^2=\mu_2\|\nabla\Re(q)\|^2\nn\\
\label{Wcdef}
=&\,\mu_2\sum_{l=1}^3\|\pal\Re(q)\|^2,\\
\label{Vdef}
V(\kappa) \,:=&\, \rho\kappa^2.
\end{align}
In (\ref{Wstdef}), (\ref{Wcdef}),
$\mu_2:=\frac\mu2L_c^2$ with the internal length scale $L_c>0$,
the Cosserat couple modulus $\mu_c>0$, and $\lambda>0$, $\mu>0$ are the
Lam{\'e} parameters; $\pal:=\frac{\partial}{\partial x_l}$, $1\le l\le3$ for
short; $\rho>0$ is a constant.
In (\ref{Wstdef}), $\sym(A):=\frac12(A+A^t)$, $\skw(A):=\frac12(A-A^t)$ denote
the symmetric and skew-symmetric part of a tensor $A$, respectively;
$\tr(A):=\sum_iA_{ii}$ is the trace operator,
$\|A\|:=\sqrt{\tr(A^tA)}$ the Frobenius matrix norm; 
$\mathbf{u}\!\cdot\!\mathbf{v}:=\sum_{i=1}^3u_iv_i$
is the inner product in $\R^3$, $\Id$ the real $3\times3$ identity matrix.
For $A,B\in\R^{3\times3}$,
$A\!:\!B:=\tr(A^tB)=\sum_{i,j=1}^3A_{ij}B_{ij}$
denotes the inner product in $\R^{3\times3}$.
For a general introduction to tensor calculus in plasticity, we recommend
\cite{HR99,Lub08}.

Applying ideas from \cite{FF95}, see also \cite{OR99}, the time evolution of the
deformed material can be computed by a sequence of minimization problems for the
mechanical energy. If $h>0$ is a fixed time step, for given
$(\Fp^0,\kappa^0)$ of the previous time step, the values of
$(\phi,\Re,\Fp,\kappa)$ need to be calculated at time $t+h$.
Let $P:=\Fp^{-1}$ be the plastic backstress, and $P^0:=({\Fp}^0)^{-1}$.
Then the approximations
\[ d_t^h(\Fp):=\frac{\Id-P^{-1}P^0}{h},\qquad \pat^h\kappa:=
\frac{\kappa-\kappa^0}{h} \]
of the time derivatives are used. Other forms of time integrators are discussed
in \cite{WA90}. We obtain the minimization problem
\begin{align}
\cE(\phi,q,\Fp,\kappa):=\io& \Wst(\Ue(\phi,q,\Fp))+\WWc(\Ke(q))+V(\kappa)
\nn\\[-2mm]
&+\! \Lambda\big(1\!-\!|q|^2\big)^2\!-\!\fext\cdot\phi\!-\!\Mext\!:\!\Re(q)\nn\\
& +hQ^*(d_t^h(\Fp),\pat^h\kappa)\ddx-\int\limits_{\Gamma_\mathfrak{t}}
\vec{\mathfrak{t}}\!\cdot\!\phi\dS\nn\\[-2mm]
\label{Emin0}
& \hspace*{-1em} -\!\int\limits_{\Gamma_C}M_\mathfrak{t}\!:\!\Re(q)\dS\to\min
\end{align}
subject to the initial and Dirichlet boundary conditions
\begin{equation}
\label{BC}
\begin{aligned}
& \phi(x,0)=x,\quad\kappa(\cdot,0)=\kappa^0\qquad\mbox{in }\Omega,\\
& \phi=g_{_D},\quad q=q_D\hspace*{56pt}\mbox{on }\Gamma_D
\end{aligned}
\end{equation}
with fixed Dirichlet boundary data $q_D$ and $g_D$. As is typical of a
variational theory, the functional $\cE$ represents the total mechanical
energy of the system minus the ground state energy.
In (\ref{Emin0}), (\ref{BC}), $\Gamma_D\subset\partial\Omega$ is that part of
the boundary where Dirichlet conditions are applied; $\Gamma_\mathfrak{t}$
is the part of the boundary where traction boundary conditions apply;
$\Gamma_C\subset\partial\Omega$ the boundary where surface couples are applied.
It must hold $\Gamma_D\cap\Gamma_\mathfrak{t}=\emptyset$,
$\Gamma_D\cap\Gamma_C=\emptyset$. For simplicity, we assume from now on
$\Gamma_D=\partial\Omega$ and $\Gamma_\mathfrak{t}=\Gamma_C=\emptyset$.

In (\ref{Emin0}), the term $\vL(|q|^2-1)^2$ ensures the validity of the
constraint $|q|=1$ in $\Omega$, where $\vL>0$ is a constant.
By $\fext=\fext(t)$, $\Mext=\Mext(t)$, the external volume force density and
external volume couples are specified, respectively. The term
$hQ^*(d_t^h(\Fp),\pat^h\kappa)$ is the dissipated mechanical energy
in the time interval from $t$ to $t+h$. It is the Legendre-Fenchel dual
\begin{equation}
\label{LFdual}
Q^*(\Fp,\kappa):=\sup_{(X,\xi)}\big\{X:\Fp+\xi\kappa-Q(X,\xi)\big\}
\end{equation}
of the plastic potential
\[ Q(X,\xi):=\left\{\begin{array}{ll}0,& \;\mbox{for }Y(X,\xi)\le0,\\
\infty,& \;\mbox{else,}\end{array}\right. \]
where $Y\le0$ is the yield function with $Y=0$ indicating plastic flow.
In case of the van Mises condition,
\[ Y(\sigma,\xi):=\|\mathrm{dev}\sym\sigma\|-\sigma_Y-\xi \]
with $\mathrm{dev}\sigma:=\sigma-\frac{1}{3}\Id$ the deviatoric part of
$\sigma$.
The above formulas establish a rate-independent theory where the material
responds immediately (infinitely fast) to applied forces.

As a result of plastic deformation due to structural changes within the
material like the increase of immobilized dislocations inside the lattice
structure, hardening occurs, \cite{CGG06,BL06}.
It is assumed throughout the text that plastic deformations only occur along
one a-priori given material-dependent single-slip system,
specified by a normal vector $\mathbf{n}$ and a slip vector $\mathbf{m}$
with $|\mathbf{m}|=|\mathbf{n}|=1$ and $\mathbf{m}\!\cdot\!\mathbf{n}=0$,
see \cite{Gurtin02}.

The real parameter $\gamma$ determines the plastic slip and the plastic
deformation tensor by
\begin{equation}
\label{Fpdef}
\Fp=\Fp(\gamma):=\Id+\gamma\,\mathbf{m}\!\otimes\!\mathbf{n}.
\end{equation}
Formula~(\ref{Fpdef}) is obtained from
$\dot{\Fp}=\dot{\gamma}\,\mathbf{m}\otimes\mathbf{n}$ by integration from
the initial state $\Fp(t=0)=\Id$ to time $t$.

In contrast to \cite{Ble13}, we restrict here to the case of one slip system,
by leaving the multislip case for future work.

As can be checked, \cite{CHM2002}, the dissipated energy satisfies the
relationship
\begin{equation}
\label{Qstardef}
Q^*(\dot{A},\dot{k})=\left\{\!\begin{array}{l@{\quad}l}
\sigma_Y|\dot{\gamma}|, & \mbox{if }\dot{A}=\dot{\gamma}\,
\mathbf{m}\!\otimes\!\mathbf{n}\mbox{ and }|\dot{\gamma}|+\dot{k}\le0,\\
\infty, & \mbox{else.} \end{array}\right.
\end{equation}
As is well known, plastic deformations always occur on the boundary of the
set of feasible deformations. Consequently, see \cite{Ble13}, the constraint
$|\gamma-\gamma^0|+\kappa-\kappa^0\le0$ appearing in the
definition of $Q^*$ has to be satisfied with equality, leading to
\begin{equation}
\label{kapdef}
\kappa=-|\gamma-\gamma^0|+\kappa^0,
\end{equation}
which allows us to define $\kappa$ as a function of $\gamma$, $\gamma^0$, and
$\kappa^0$. Plugging in (\ref{kapdef}) in $V(\kappa)$ and dropping an
inconsequential constant $\rho(\kappa^0)^2$,
we end up with the optimization problem
\begin{equation}
\label{Emin1}
\begin{aligned}
\cE(\phi,q,\gamma):=\io\!\Big[&\Wst(\Re^t(q)D\phi\Fp(\gamma)^{-1})+\WWc(q)
\\[-1.5mm]
&+\vL(|q|^2-1)^2-\fext\cdot\phi-\Mext:\Re(q)\\
& +\rho\big(\gamma\!-\!\gamma^0\big)^2+|\gamma\!-\!\gamma^0|\big(
\sigma_Y\!-\!2\rho\kappa^0\big)\Big]\ddx\to\min
\end{aligned}
\end{equation}
subject to the initial and boundary conditions (\ref{BC}) for
$\Gamma_D=\partial\Omega$.

The functional $\cE$ in (\ref{Emin1}) coincides with the one in
\cite{Ble14} except for the new term $\Lambda\big(|q|^2-1\big)^2$ and
the parameterization (\ref{para1}) instead of (\ref{para2}) for the
micro-rotations.

For a fixed discrete time step $h>0$ and
known $(\gamma^0,\kappa^0)$ at time $t$, the new $(\phi,q,\gamma)$
representing values at time $t+h$ are calculated from (\ref{Emin1}). Finally,
the new $\kappa$ is computed from (\ref{kapdef})
and $(\gamma,\kappa)$ become the initial values of the next time step.

If the material is initially free of dislocations, $\kappa(\cdot,0)=0$,
the hardening law (\ref{kapdef}) implies $\kappa(t+h)\le\kappa(t)\le0$ for all
times $t$. Hence, $-2\rho\kappa^0\ge0$ in (\ref{Emin1}) represents
the increase of the yield stress $\sigma_Y$ due to stored dislocations.

\section{An application of the Euler-Rodrigues formula}
\label{secER}
Following the classical notation in \cite{Hamilton,Ebbing}, let
\begin{align*}
\H &:= \mathrm{span}_\R\{1,\bbi,\bbj,\bbk\}\\
&= \big\{q=q_0+q_1\bbi+q_2\bbj+q_3\bbk\;\big|\;q_0,q_1,q_2,q_3\in\R\big\}
\end{align*}
denote the {\it space of quaternions}, where the quaternion imaginary units
satisfy $\bbi^2=\bbj^2=\bbk^2=\bbi\bbj\bbk=-1$. Let
\[ \Hp:=\{q=q_0+q_1\bbi+q_2\bbj+q_3\bbk\in\H\;|\;q_0=0\} \]
be the space of {\it pure quaternions} and
\begin{equation}
\label{qhat}
q=q_0+\hq:=q_0+q_1\bbi+q_2\bbj+q_3\bbk.
\end{equation}
The set $\H$ is equipped with the multiplication (for $p,q\in\H$)
\begin{equation}
\label{Hmult}
pq:=p_0q_0-\hp\cdot\hq+p_0\hq+q_0\hp+\hp\times\hq,
\end{equation}
where $\hp\cdot\hq:=p_1q_1+p_2q_2+p_3q_3$ specifies as above the inner
product and $\hp\times\hq$ the vector product of $\R^3$, respectively.
In general, $pq\not=qp$, so $\H$ is an associative, non-commutative algebra.
Let $\ovq:=q_0-\hq$ be the {\it conjugate} of $q$ and
\begin{equation}
\label{qnorm}
|q|:=\big(q\ovq\big)^{1/2}=\big(\ovq q\big)^{1/2}
=\big(q_0^2+q_1^2+q_2^2+q_3^2\big)^{1/2}
\end{equation}
be the {\it modulus} of $q$. By Formula~(\ref{Hmult}),
$q\in\H^*:=\H\setminus\{0\}$ possesses the multiplicative inverse
$q^{-1}=\frac{\ovq}{|q|^2}$. Let
\[ \mathrm{so}(3):=\{\omega\in\R^{3\times3}\;|\;\omega^t=-\omega\} \]
be the Lie algebra of $\SO(3)$. The {\it alternating skew tensor} 
$\ve:\Hp\to\mathrm{so}(3)$ is defined by
\begin{equation}
\label{epsdef}
\ve(\hq):=\left(\begin{array}{ccc} 0 & -q_3 & q_2\\ q_3 & 0 & -q_1\\
-q_2 & q_1 & 0 \end{array}\right).
\end{equation}
Evidently,
\begin{equation}
\label{epsprop}
\ve(\hq)v=\hq\times v\qquad\mbox{for }v\in\R^3\simeq\Hp.
\end{equation}
By direct inspection, it is straightforward to verify that for every $q\in S^3$
\begin{equation}
\label{R2def}
\Re(q)v:=qv\ovq
\qquad\mbox{for }v\in\R^3\simeq\Hp
\end{equation}
defines a rotation in $\SO(3)$.
Using (\ref{Hmult}), this leads to
\begin{equation}
\label{para1b}
\Re(q)\,=\,(2q_0^2-|q|^2)\Id+2\hq\otimes\hq+2q_0\ve(\hq).
\end{equation}
Plugging in the above definitions, this coincides with Formula~(\ref{para1}).

The mapping $\Re$ thus introduced has the properties
\[ \Re(1)=\Id,\qquad\Re(\ovq)=\Re(q)^t,\qquad\Re(pq)=\Re(p)\Re(q) \]
and is therefore an algebra-homomorphism. It is a double cover of $\SO(3)$,
especially it is non-unique, since
\begin{equation}
\label{Rdouble}
\Re(q)=\Re(-q)\qquad\mbox{for }q\in S^3.
\end{equation}
In comparison, the parameterization (\ref{para2}) breaks down for
$\alpha_2=\frac\pi2$, in which case $\alpha_1$ and $\alpha_3$ denote a rotation
around the same axis. In summary, both (\ref{para1}) and (\ref{para2}) set up
rivaling charts on the manifold $\SO(3)$ which have certain disadvantages when 
used globally.
\\
Formula~(\ref{para1}) can be used to interpolate between rotations and allows
to introduce a distance in $\SO(3)$, see, e.g., \cite{DKL98}. This is a
prerequisite to studying surface energies between grains or particles
of different orientations, \cite{Ble17}.
\\
For $x\in\R^3$ and a quaternion field $q=q(x)$, the $m$-th material
{\it curvature vector} or {\it Darboux vector} is given by
\begin{equation}
\label{Kmdef}
\Ke^m(q):=2\ovq\pam q\in\Hp,\qquad1\le m\le3.
\end{equation}
The following lemma computes the derivatives of $\Re(q)$ and $\Ke(q)$ in $\H$
with $|q|=1$.
\begin{lemma}[Lie Derivatives of $\Re$ and $\Ke^m$]
\label{lem1}
Let $q=q(x):\R^3\to S^3$ and $1\le l,m\le 3$. Then
\begin{align}
\label{Lie1}
\pal\Re(q) &= \Re(q)\ve(\Ke^l(q)),\\
\label{Lie2}
\pal\Ke^m(q) &= 2\ovq\big[\pal\pam q-\pal q\ovq\pam q\big].
\end{align}
\end{lemma}
\begin{proof}
An elementary proof of (\ref{Lie1}) can be found in \cite{Kuipers99},
Chapter~11. The following proof is a modification of an argument in
\cite{LL09}.
Let $v\in\R^3\simeq\Hp$ and let $w\in\R^3$ denote various changing vectors.
Then it holds
\begin{alignat*}{2}
\ve(\Ke^l(q))v &= \ve(2\ovq\pal q)v &&\mbox{by Eqn.~(\ref{Kmdef})}\\
&= 2\ovq\pal q\times v &&\mbox{by Eqn.~(\ref{epsprop})}\\
&= 2\ovq\pal q v &&\mbox{by Eqn.~(\ref{Hmult})}\\
&= 2\widehat{\ovq\pal qv} &&\mbox{by Eqn.~(\ref{qhat})}\\
&= \ovq\pal qv-\overline{\ovq\pal q v} &&\mbox{since }w-\overline{w}
=2\widehat{w}\\
&= \ovq\pal qv+v\pal\ovq q &&\mbox{since }\overline{v}=-v\\
&= \ovq(\pal qv\ovq+qv\pal\ovq)q\quad &&\mbox{since }\ovq q=|q|^2=1\\
&= \ovq(\pal(qv\ovq))q &&\mbox{since }\pal v=0\\
&= \ovq(\pal\Re(q)v)q &&\mbox{by Eqn.~(\ref{R2def})}\\
&= \Re(q)^t\pal\Re(q)v &&\mbox{since }(\Re(q)w)^t=\ovq wq.
\end{alignat*}
As this is true for every $v\in\R^3\simeq\Hp$, this shows
\[ \ve(\Ke^l(q))=\Re(q)^t\pal\Re(q). \]
Multiplication with $\Re(q)$ from the left yields (\ref{Lie1}).\\
In order to show (\ref{Lie2}), multiplying (\ref{Kmdef}) with $q$
from the left yields
\[ 2\pam q=q\Ke^m(q). \]
Consequently,
\[ 2\pal\pam q=\pal q\Ke^m(q)+q\pal\Ke^m(q) \]
or equivalently
\[ q\pal\Ke^m(q)=2\pal\pam q-\pal q\Ke^m(q). \]
Multiplication of this identity with $\ovq$ from the left leads to
\[ \pal\Ke^m(q)=2\ovq\pal\pam q-\ovq\pal q\Ke^m(q). \]
With (\ref{Kmdef}), this shows (\ref{Lie2}). \qed
\end{proof}
\vspace*{2mm}
Applying the results of Lemma~\ref{lem1} to $\WWc$, it holds by
Eqns.~(\ref{Lie1}) and (\ref{epsdef}),
\begin{align}
\WWc(q) &= \mu_2\sum_{l=1}^3||\pal\Re(q)||^2=\mu_2\sum_{l=1}^3||\Re(q)
\ve(\Ke^l(q))||^2\nn\\
&=\mu_2\sum_{l=1}^3||\ve(\Ke^l(q))||^2\nn\\
&= 2\mu_2\sum_{l=1}^3\Big[({\Ke}_1^l(q))^2+({\Ke}_2^l(q))^2
+({\Ke}_3^l(q))^2\Big]\nn\\
\label{WcKe}
&=2\mu_2\sum_{l=1}^3|\Ke^l(q)|^2.
\end{align}
For the first derivative, using (\ref{Kmdef}) and (\ref{Lie2}), this results in
\begin{align}
\pam\WWc(q) &= 4\mu_2\sum_{l=1}^3\widehat{\pam\Ke^l(q)}\cdot\widehat{\Ke^l(q)}
\nn\\
\label{dWc}
&= 16\mu_2\sum_{l=1}^3\big[\ovq(\pam\pal q-\pam q\ovq\pal q)\big]\cdot\big[
\ovq\pal q\big].
\end{align}

\section{Preconditioning}
\label{secprec}
When implementing the L-BFGS method for the Cosserat problem (\ref{Emin1}),
frequently situations are encountered where the algorithm
requires many iterations to converge. Also it may happen that the iteration is
stopped before a correct minimizer has been reached. Therefore, in this section,
certain modifications of the L-BFGS algorithm are discussed.
It is noteworthy that this does not only increase the speed of the code, but
may be an essential step to correctly compute the minimizers.\\
Starting point is the minimization problem (\ref{Emin1}) written as
\begin{equation}
\label{Emin2}
\cE(x)\to\min,
\end{equation}
where $x\!\in\!\R^D$ corresponds to a spatial discretization of
$(\phi,\!q,\!\gamma)$ by finite elements or finite differences.
The L-BFGS algorithm is a quasi-Newton method and constructs a minimizing
sequence $(x_k)_{k\in\N}\subset\R^D$ by setting
\begin{equation}
\label{liter}
\begin{aligned}
d_k &:= -H_k\nabla\cE(x_k),\\
x_{k+1} &:= x_k+\alpha d_k.
\end{aligned}
\end{equation}
Here, $H_k$ approximates the inverse Hessian $(D^2\cE(x_k))^{-1}$ and is
constructed from rank-one updates, $d_k$ is a descent direction, and
$\alpha\in\R$ is a parameter computed by a linesearch algorithm.
The iteration~(\ref{liter}) stops if for chosen small $\ve_0>0$
\begin{equation}
\label{lstop}
|\nabla\cE(x_k)|<\ve_0\max\{1,|x_k|\}.
\end{equation}
Letting
\begin{align*}
s_{k-1} &:=x_k-x_{k-1},\\
y_{k-1} &:=g_k-g_{k-1}:=\nabla\cE(x_k)-\nabla\cE(x_{k-1}),
\end{align*}
the {\it BFGS-update} is given by
\begin{align}
H_k =& H_{k-1}+\Big(\frac{y_{k-1}^tH_{k-1}y_{k-1}}{y_{k-1}^ts_{k-1}}+1\Big)
\frac{s_{k-1}s_{k-1}^t}{y_{k-1}^ts_{k-1}}\nn\\
\label{Hup1}
&-\frac{1}{y_{k-1}^ts_{k-1}}\Big[
s_{k-1}y_{k-1}^tH_{k-1}+H_{k-1}y_{k-1}s_{k-1}^t\Big]\\
=& \big(\Id-\rho_{k-1}s_{k-1}y_{k-1}^t)H_{k-1}
\big(\Id-\rho_{k-1}y_{k-1}s_{k-1}^t)\nn\\
& +\rho_{k-1}s_{k-1}s_{k-1}^t\nn\\
=:& V_{k-1}^tH_{k-1}V_{k-1}+\rho_{k-1}s_{k-1}s_{k-1}^t\nn\\
=& \big(V_{k-1}^t\ldots V_0^t\big)H_0\big(V_0\ldots V_{k-1}\big)
\label{Hup2}\\
& +\sum_{l=1}^{k-1}(V_{k-1}^t\ldots V_l^t)s_{l-1}s_{l-1}^t
\big(V_l\ldots V_{k-1}\big)+\rho_{k-1}s_{k-1}s_{k-1}^t\nn
\end{align}
with $\rho_l:=\frac{1}{y_l^ts_l}$ and $V_l:=\Id-\rho_ly_ls_l^t$.

\noindent In the limited-memory variant of (\ref{Hup1}),
the matrices $H_k$ are not
stored explicitly. Instead, given a small number $m\in\N$ and vectors
$s_0,\ldots,s_{m-1}$, $y_0,\ldots, y_{m-1}$, the multiplication
\[ H_k\nabla\cE(x_k) \]
is carried out by the two-loop iteration, see \cite{Nocedal80},\cite{LN89},
\begin{align}
& g_k:=\nabla\cE(x_k)\nn\\
& \mbox{FOR }i=m-1,\ldots,0\nn\\
& \qquad\alpha_i:=\rho_is_i^tg_k\nn\\
& \qquad g_k:=g_k-\alpha_iy_i\nn\\[3mm]
\label{Hmul}
& r_k:=H_k^0g_k\\[3mm]
& \mbox{FOR }i=0,\ldots,m-1\nn\\
& \qquad\beta_k:=\rho_iy_i^tr_k\nn\\
& \qquad r_k:=r_k+(\alpha_i-\beta_k)s_i\nn\\
& H_k\nabla\cE(x_k):=r_k\nn.
\end{align}
The first FOR-loop of the above scheme for determining
$r_k=H_kg_k$ computes and stores $\big(V_l\ldots V_{m-1}\big)g_k$ for
$0\le l\le m-1$. After carrying out the multiplication (\ref{Hmul}),
the second FOR-loop then computes (\ref{Hup2}).\\
The above scheme is considered one of the most effective update formulas of
numerical analysis and requires only ${\mathcal O}(mD)$ operations.
The parameter $m$ is usually chosen as $3\le m\le7$, see \cite{Byrd94},
and increasing $m$ further does not improve the quality of the update.
\\ In (\ref{Hmul}), for each iteration step $k$, one is free to pick $H_k^0$.
In the original implementation of the algorithm, in order to reduce the
condition numbers of $H_k$, the diagonal is scaled with the
{\it Cholesky factor} $\delta_k$, \cite{OL74},
\begin{equation}
\label{Cholesky}
H_k^0=\delta_k\Id,\qquad \delta_k:=\frac{s_{k-1}^ty_{k-1}}{y_{k-1}^ty_{k-1}}.
\end{equation}
Instead, another matrix or non-linear scheme such as a fixed point iteration
may be used in place of $H_k^0$ in (\ref{Hmul}) such that ideally,
$H_k^0\sim D^2\cE(x_k)$.
\\In order to find an efficient preconditioning method, it is helpful to study
the particular features of the Cosserat functional $\cE$. From physical insight
and numerical investigations, it is evident that the hardest part
in solving (\ref{Emin1}) is the computation of the optimal rotations, i.e.
finding the quaternion field $q$. Therefore, the following two-step
strategy for the solution of one time-step is effective:
\\
\noindent\underline{\bf Step~1 (Predictor)}: Fix $(\phi,\gamma)$.\\
Solve with the L-BFGS-method the optimization problem
\[ \cE_{\phi,\gamma}(q)\to\min. \]

\noindent\underline{\bf Step~2 (Corrector)}: Solve with the L-BFGS-method the
full problem (\ref{Emin2}). Pick the solution $q_\mathrm{opt}$ of Step~1 as
initial values for $q$.

\vspace*{2mm}
Typically, the solution of Step~1 is very fast in comparison to Step~2 since
far less variables need to be optimized and the complicated dependence of $q$
on $(\phi,\gamma)$ is eliminated.
Step~1 provides a reasonable approximation to the solution of the full problem
(\ref{Emin2}). In the conducted tests, the combined numerical costs for solving
Step~1 and Step~2 turned out significantly lower than for solving the original
minimization problem directly in one pass with the un-precon\-ditioned
L-BFGS method. This is discussed below in more detail.

In Step~1, $(\phi,\gamma)$ is fixed with data of the previous
time step. At the first time step, $\gamma$ is loaded with the initial values
$\gamma^0$ and $\phi$ is initialized with an extension of the
boundary values $g_D$ in $\overline{\Omega}$
that satisfies the Cauchy-Born rule.

Both Step~1 and Step~2 are preconditioned.
In Step~1, a special preconditioning matrix $Z$ replacing $H_k^0$ is chosen
that resembles the common discretization of the Laplace operator on structured
grids. Step~2 is preconditioned with the final converged matrix $H_k$ computed
in Step~1. As this matrix is obtained from a L-BFGS-procedure, it has a
data-sparse representation by vectors $(s_0,y_0),\ldots, (s_{m-1},y_{m-1})$.

In order to derive the preconditioning-matrix $Z$ of Step~1, recall the
computation of the total curvature energy by finite differences in 3D
\begin{equation}
\label{NC}
\io\WWc(q)\ddx\approx\frac{w}{8}\,\sum_{i=0}^{d_1}\sum_{j=0}^{d_2}
\sum_{k=0}^{d_3}N_{ijk}\WWc(q(y_{ijk}))
\end{equation}
used in \cite{Ble15}, where $N_{ijk}\in\N$ are
numerical weights derived from a Newton-Cotes formula,
$y_{ijk}\in\overline{\Omega}$ are points of the numerical mesh with equal
spacings
\begin{equation}
\label{xyz}
\eta_1:=\frac{L_1}{d_1},\qquad \eta_2:=\frac{L_2}{d_2},\qquad
\eta_3:=\frac{L_3}{d_3},
\end{equation}
$\Omega=(0,L_1)\times(0,L_2)\times(0,L_3)$ is assumed, $d_l\in\N$ is the
number of discretization points in direction $l$, $l=1,2,3$, and
$w:=\eta_1\eta_2\eta_3$ is an integration factor.

Since for the preconditioning matrix only a reasonably good approximation of the
second derivative is needed, in the following $N_{ijk}=8$ is assumed (the value
of $N_{ijk}$ in $\overline{\Omega}\setminus\partial\Omega$). First, let
\[ W_c(q):=2\mu_2|\pax q|^2. \]
Then, by a straightforward computation, for fixed subscripts
$0\le I\le d_1$, $0\le J\le d_2$, $0\le K\le d_3$ and fixed component
$0\le b\le3$ of $q$,
\begin{align}
\frac{\partial}{\partial q_{IJK}^b}&\io\WWc(q(x))\ddx\approx
\frac{w\mu_2}{2\eta_1^2}\frac{\partial}{\partial q_b^{IJK}}
\sum_i\Big(q_{i+1,J,K}^b-q_{i-1,J,K}^b\Big)^2\nn\\
&=\frac{w\mu_2}{\eta_1^2}\sum_i\Big(q_{i+1,J,K}^b
-q_{i-1,J,K}^b\Big)\Big(\delta_{i+1,I}-\delta_{i-1,I}\Big)\nn\\
\label{entries}
&=\frac{w\mu_2}{\eta_1^2}\Big(-q_{I-2,J,K}^b
+2q_{I,J,K}^b-q_{I+2,J,K}^b\Big)
\end{align}
with the short-hand notation $q_{ijk}\equiv q(y_{ijk})$.
In the same way the second derivative
\[ \frac{\partial^2}{(\partial q_{IJK}^b)^2}\int_\Omega\WWc(q(x))\ddx \]
can be computed. Let $D_1\widehat{=}(I,J,K)$ be the line index and
$D_2\widehat{=}(I_2,J_2,K_2)$ be the column index of the 2nd derivative
matrix $Z$. Then, from (\ref{entries}),
\[ Z_{D_1,D_2}:=\frac{w\mu_2}{\eta_1^2}\left\{\!\!\begin{array}{ll}
+2, \quad\mbox{ if }D_1=D_2,\\
-1, \quad\mbox{ if }|I-I_2|=2,\\
\;\;\;0, \quad\mbox{ otherwise}.\end{array}\right. \]
Likewise, if $\WWc$ is given by (\ref{Wc2def}), then up to a
pre-factor, $Z$ is $2$ on the diagonal, equals $-1$ if $|I-I_1|=2$ or
$|J-J_1|=2$ or $|K-K_1|=2$, and is $0$ otherwise.

\vspace*{2mm}
In the implementation, $Z$ is not stored explicitly. The multiplication $Zg$
for a vector $g\in\R^D$ is carried out by exploiting the band structure of
$Z$.

\section{Numerical tests}
\label{secnum}
Subsequently, different algorithms for the solution of (\ref{Emin1})
are investigated. First, the following general remarks are in place.

\begin{remark}
\label{rem1}
Following \cite{Ble15}, for small $\ve>0$, in (\ref{Emin1}) the modulus
$|\cdot|$ is replaced by
\[ r_\ve(x):=\left\{\!\!\begin{array}{rl} x, &\quad x>\ve,\\ x^2/\ve, & \quad
-\ve\le x\le+\ve,\\
-x, & \quad x<-\ve. \end{array}\right. \]
This removes the singularity at the origin and allows the application of
Newton's method.
\end{remark}

\begin{remark}
\label{rem2}
Since the quasi-Newton method applied in this article computes variations of $q$
that are not in $S^3$, the parameterization (\ref{para1}) is not applicable
unmodified in the numerical code. Instead, the mapping
\begin{equation}
\label{para1c}
{\footnotesize{
\wRe(q)}}
{\footnotesize{\!:=\!\frac{1}{|q|^2}\!\!\!
\left(\!\!\begin{array}{c@{\quad}c@{\quad}c}
q_0^2\!+\!q_1^2\!-\!q_2^2\!-\!q_3^2 & 2(q_1q_2\!-\!q_0q_3)&
2(q_1q_3\!+\!q_0q_2)\\
2(q_1q_2\!+\!q_0q_3)& q_0^2\!-\!q_1^2\!+\!q_2^2\!-\!q_3^2 &
2(q_2q_3\!-\!q_0q_1)\\
2(q_1q_3\!-\!q_0q_2)& 2(q_2q_3\!+\!q_0q_1)& q_0^2\!-\!q_1^2\!-\!q_2^2\!+\!q_3^2
\end{array}\!\!\right)}}
\end{equation}
is used which is defined for all $q\in\R^4\setminus\{0\}$. When minimizing
$\cE_\ve$, due to the term $\vL\big(|q|^2-1\big)^2$, the computed
optimal $q$ lies (approximately) in $S^3$.
\end{remark}

\begin{remark}
\label{rem3}
All plastic deformations considered in this section satisfy $\det(\Fp)=1$.
Hence the plastic deformations preserve the volume.
\end{remark}
\subsection{Comparison of the parameterizations by Euler angles and
Euler-Rodrigues formula}
\label{subcomp}
The quaternion-based algorithm, due to its additional component in the
representation of $\Re$, requires about $14\%$ more computer memory.
Table~\ref{tab1} has the exact figures for different spatial resolutions.
Let {\it Algorithm~1} denote the algorithm of \cite{Ble15} which is
based on finite differences in 3D, the L-BFGS method, Euler angles, and the
curvature energy (\ref{Wcdef}), {\it Algorithm~2a} be the analogous
quaternion-based algorithm that solves (\ref{Emin1}); finally {\it Algorithm~2b}
be identical to Alg.~2a, but with the simplified curvature energy
\begin{equation}
\label{Wc2def}
\widetilde{W_c}(q):=2\mu_2\sum_{l=1}^3|\pal q|^2.
\end{equation}
This choice is motivated by the fact that Euler angles permit to write
(\ref{Wcdef}) as
\begin{equation}
\label{Wcalpha}
W_c(\alpha)=2\mu_2\sum_{l=1}^2|\pal\alpha|^2,
\end{equation}
see Eqn.~(\ref{WcKe}) or \cite{Ble14}.
As the numerical costs for computing (\ref{Wc2def}) and (\ref{Wcalpha}) are
very similar, this permits an unbiased comparison of the two parameterizations.

In \cite{Ble14}, a class of 3D analytic solutions to (\ref{Emin1}) is
calculated for an ultra-soft material with $\sigma_Y=\rho=0$ subject to the
boundary conditions
\begin{equation}
\label{shearBC}
D\phi(t)=\Id+\beta(t)\mathbf{m}\!\otimes\!\mathbf{n}
\qquad\mbox{on }\partial\Omega.
\end{equation}
This represents a simple shear problem for prescribed values $\beta(t)\in\R$.
The Cauchy-Born rule is valid here and (\ref{shearBC}) is satisfied
in $\overline{\Omega}$.

The above test constitutes a benchmark problem. The following simulation
compares the performance and speed of convergence for both Alg.~1 and Alg.~2.
The stopping criterion is (\ref{lstop}) with $\ve_0:=10^{-7}$.

\vspace*{-5mm}
\begin{align*}
& \mbox{{\bf Parameters (Benchmark test)}:}\\[-1mm]
& \Omega=(0,1)^3,\; t\in[0,1],\;\mu=10^4,\;\mu_c=2\cdot10^4,\\[-1mm]
&
\mu_2:=\mu\frac{L_c^2}{2}=100,\;\lambda=10^3,\;
\rho=\sigma_Y=\fext=0,\; \Mext=\mathbf{0},\\[-1mm]
& \mathbf{m}=(1,0,0)^t,\;\mathbf{n}=(0,1,0)^t,\;
\beta(t)=0.25*t,\; h=0.1,\\[-1mm]
&\; \ve=10^{-4},\; \vL=20,\;q_D=(1,0,0,0).\\
& \mbox{\it Initial values: }\phi_0\equiv\Id,\;\kappa^0=\gamma^0=0
\mbox{ in }\Omega.
\end{align*}


\noindent{\bf Results}: $\gamma(\cdot,t)=\beta(t)$, $\Re=\Ue=\Id$,
$\Wst=\WWc=0$ in $\overline{\Omega}$,

\hspace*{34pt}$\phi(x,t)=(x_1+\beta(t)x_2,x_2,x_3)$ in $\Omega$, i.e.
the validity of the Cauchy-Born rule.

\vspace*{3mm}
Table~\ref{tab2} summarizes the required number of iterations and computation
times for all variants. The stopping criterion is (\ref{lstop}) with
$\ve_0:=10^{-7}$. As can be seen, Alg.~2b requires about $20\%$ less iterations,
Alg.~2a about $10\%$ less iterations than Alg.~1.
This behavior is typical. In our numerical tests, the
quaternion-based algorithms reveal superior convergence.
Table~\ref{tab3} illustrates the deviation of the numerical
solution from the constraint $|q|=1$.
\subsection{The effect of preconditioning}
\label{subpc} 
This section conducts numerical tests of the preconditioning strategy presented in Section~\ref{secprec}. While for large values of the stop parameter
$\ve_0$ the code usually converges after a small number of iterations,
preconditioning becomes mandatory when $\ve_0$ is chosen small.
Fig~\ref{fig1} demonstrates that reducing $\ve_0$ may go along
with an exponential increase of the number of iterations.
Simultaneously, fine properties of the physical solution may be
missed when $\ve_0$ is set too large, cf. also Table~\ref{tab3}.
The following bending problem of a 3D rod, see \cite[Eqn.~(27)]{Ble13},
serves as a test problem. For given $\beta(t)$ as in (\ref{shearBC}), $\phi$ at
$\partial\Omega$ is prescribed by
\begin{equation}
\label{phiwarp}
g_D^\mathrm{bend}(x_1,x_2,x_3,t):=\left(\begin{array}{c}x_1\\
x_2+\frac{2L_1}{\pi}\Big[\sin\Big(\frac{3\pi}{2}+\frac\pi2\frac{x_1}{L_1}\Big)+1
\Big]\beta(t)\\ x_3\end{array}\right).
\end{equation}
In order to determine the boundary conditions on $q$, let
\[ R_D^\mathrm{bend}:=\mathrm{polar}(Dg_D^\mathrm{bend}\Fp^{-1}), \]
where $\mathrm{polar}(\cdot)$ is the polar decomposition, computed with the
algorithm in \cite{Dongarra79}. Then set
\[ q=q_D^\mathrm{bend}\mbox{ on }\partial\Omega, \]
where $\Re(q_D^\mathrm{bend})=R_D^\mathrm{bend}$ and $q_D^\mathrm{bend}$ is
computed from $R_D^\mathrm{bend}$ with the algorithm in
\cite{Shoe85}. 

\begin{align*}
& \mbox{{\bf Parameters (Bending problem)}:}\\[+1mm]
& \Omega=(0,5)\!\times\!(0,1)\!\times\!(0,2),\;t\in[0,1],\;
\lambda=\mu=0.025,\\&\;\mu_c=0.4,\;\mu_2=0.02;\;\rho=\sigma_Y=\fext=0,\;
\Mext=\mathbf{0},\\
& \mathbf{m}=(1,0,0)^t,\;\mathbf{n}=(0,1,0)^t,\;
\beta(t)=0.25*t,\;h=0.1,\\&\;\ve=10^{-4},\;\vL=20,\;
\WWc(q)=2\mu\sum_{l=1}^3|\pal q|^2,\\[-2mm]
& \mbox{\it Initial values: }\phi_0\equiv\Id,\;\kappa^0=\gamma^0=0
\mbox{ in }\Omega.\\
& \mbox{\it Boundary values: }\phi=g_D^\mathrm{bend},\;
q=q_D^\mathrm{bend}\mbox{ on }\partial\Omega.
\end{align*}

\vspace*{2mm}
\noindent{\bf Results}:
$\gamma(x,t)=\sin(\frac\pi2\frac{x_1}{L_1})\beta(t)$, $q=(1,0,0,q_4)$,
$\Ue=\Id$,\\
\hspace*{36pt}$\Wst\equiv0$, $\phi=g_D^\mathrm{bend}$ in $\overline{\Omega}$.

\vspace*{5mm}
Table~\ref{tab4} compares the numerical costs for solving the first time
step of the bending problem with the original L-BFGS-algorithm (where
$H^0_k$ is defined by (\ref{Cholesky})) and with the preconditioned two-step
L-BFGS-algorithm of Section~\ref{secprec} when $\ve_0:=10^{-11}$.
Again, this behavior is typical. In our numerical tests, the two-step
preconditioner leads to a significant speed-up, often accompanied
with increased precision.

\section{Discussion}
\label{secdis}
In this paper, a parameterization by quaternions is applied to a strongly
non-linear finite-strain Cosserat model of plastic materials, 
possibly with microstructure. Despite
increased memory requirements, in the conducted numerical tests the
quaternion-based algorithm needed less iterations and converged faster.
As main result, a novel two-level preconditioning scheme is proposed
that exploits the physical properties of the Cosserat model.
The preconditioner solves a simplified problem for $\Re$ with fixed
$(\phi,\gamma)$ which represents the most complicated step in computing a global
minimizer of $\cE_\ve$. Note that the degrees of freedom in the micro-rotations
are responsible for the occurrence of a large number of local minima.
With this reasonably good guess for $\Re$, the preconditioned algorithm
is eventually able to succeed to a global minimizer.
The preconditioning strategy is compatible with the L-BFGS update scheme
and can be regarded as a non-linear preconditioning technique.
Numerical tests support that this scheme significantly reduces the
algorithmic costs and is essential to computing the physical solution
when high precision is required.
Similar two-step L-BFGS-algorithms may also be applicable to other classes of
problems that depend in an un-symmetrical way on its variables.
Fig.~\ref{fig2} documents a further important numerical feature:
Since the energy landscape of $\cE_\ve$ consists of many flat plateaus, the
L-BFGS-scheme stagnates for a long time with each step only slightly decreasing
$\cE_\ve$. It is unknown if and when an iteration significantly decreases
the energy. When $\ve_0$ in (\ref{lstop}) is taken too large, the algorithm may
wrongly interpret this stagnation as convergence. It would be desirable to
have analytic results on the choice of $\ve_0$, or even better an algorithm
that is capable to prevent this stagnation period.
Finally, it may be desirable to develop a specialized L-BFGS algorithm that
restricts the variations of the functional w.r.t. certain variables to
the tangent space.

We address the above generalizations and enrichments of the numerical model
presented in this study, together with the analysis of the more general case of
gradient-type plasticity and hysteretic response under general loading,
\cite{Gurtin02,CHM2002,bardellaJMPS}, to future work.
Additional future research lines will be devoted to applying the current
Cosserat model to bending-dominated lattices with plastic behavior which
exhibit arbitrarily large size effects and consist, e.g., of cubical
modules/particles connected by deformable links
or Sarrus linkages tessellating triangular lattice structures \cite{ZMP17}.
Physical models of such systems will be fabricated through AM in ductile
materials, \cite{pentamode}. These mockups will be laboratory-tested
in order to validate the accuracy of numerical simulations and to demonstrate
the presence of size effects that cannot be described through classical
continuum or homogenization theories. 
Recent results revealing metamaterial-type behaviors of the above systems,
which are related to auxetic response \cite{ZMP17} and/or high strength
effects induced by bending and twisting of the material, will be
extended to the plastic regime on accounting for a ductile response of the
background material.

\section*{Figure captions}

{\renewcommand{\arraystretch}{1.20}
\begin{table}[h!t]
\begin{center}\tabcolsep=0.45mm
\begin{tabular}{|c|c|r|r|c|}\hline
\multicolumn{1}{|c|}{\mbox{\footnotesize Parameterizations }} &
\multicolumn{1}{|c|}{\mbox{\footnotesize Resolution }} &
\multicolumn{1}{|c|}{\mbox{\footnotesize No. }} &
\multicolumn{1}{|c|}{\mbox{\footnotesize No. }} &
\multicolumn{1}{|c|}{\mbox{\footnotesize Memory }}\\
\multicolumn{1}{|c|}{\mbox{\footnotesize }} &
\multicolumn{1}{|c|}{\mbox{\footnotesize  $(d_1,d_2,d_3)$}} &
\multicolumn{1}{|c|}{\mbox{\footnotesize Unknowns}} &
\multicolumn{1}{|c|}{\mbox{\footnotesize Nodes}} &
\multicolumn{1}{|c|}{\mbox{\footnotesize  (MB)}}\\
\hline\hline
{\footnotesize Euler angles}& {\footnotesize$ (32,32,32)$}    & {\footnotesize$ 214683$}     & {\footnotesize$ 35937$}    & {\footnotesize$ 1.64$}\\ \hline
{\footnotesize Euler angles} & {\footnotesize$(64,64,64)$}    & {\footnotesize$1774907$}    & {\footnotesize$274625$}   & {\footnotesize$13.5$}\\ \hline
{\footnotesize Euler angles} & {\footnotesize$(128,128,128)$} & {\footnotesize$14436987$}   & {\footnotesize$2146689$}  & {\footnotesize$110$}\\ \hline
{\footnotesize Euler angles} & {\footnotesize$(256,256,256)$} & {\footnotesize$116462843$}  & {\footnotesize$16974593$} & {\footnotesize$889$}\\ \hline
\hline
{\footnotesize Quaternions} & {\footnotesize$(32,32,32)$}     & {\footnotesize$244476$}     & {\footnotesize$35937$}    & {\footnotesize$1.86$}\\ \hline
{\footnotesize Quaternions} & {\footnotesize$(64,64,64)$}     & {\footnotesize$2024956$}    & {\footnotesize$274625$}   & {\footnotesize$15.4$}\\ \hline
{\footnotesize Quaternions} & {\footnotesize$(128,128,128)$}  & {\footnotesize$16485372$}   & {\footnotesize$2146689$}  & {\footnotesize$126$}\\ \hline
{\footnotesize Quaternions} & {\footnotesize$(256,256,256)$}  & {\footnotesize$133044220$}  & {\footnotesize$16974593$} & {\footnotesize$1015$}\\ \hline
\end{tabular}
\end{center}
\vspace{-2mm}
\caption{\label{tab1}
Comparison between the parameterizations by Euler angles (\ref{para2}) and by
quaternions (\ref{para1}). 'No. Unknowns' is the total number of unknowns in
the discrete model, 'Memory' the total memory for storing the data in case of
$64$ bit precision, 'No Nodes' the total number of discretization
points in the finite difference mesh.}
\end{table}}


{\renewcommand{\arraystretch}{1.20}
\begin{table}[h!t]
\begin{center} \tabcolsep=0.7mm
\begin{tabular}{|c|c|c|c|c|}\hline
\multicolumn{1}{|c|}{\mbox{\footnotesize Algorithm}} &
\multicolumn{1}{|c|}{\mbox{\footnotesize $10\times10\times10$}} &
\multicolumn{1}{|c|}{\mbox{\footnotesize $30\times30\times30$}} &
\multicolumn{1}{|c|}{\mbox{\footnotesize $50\times50\times50$}} &
\multicolumn{1}{|c|}{\mbox{\footnotesize $70\times70\times70$}}\\
\hline\hline

{\footnotesize Alg.~1}  & {\footnotesize$359\;(0.91$s$)$}  & {\footnotesize$1022\;(119$s$)$} & {\footnotesize $1228\;(723$s$)$}
& {\footnotesize$1383\;(2040$s$)$} \\\hline
{\footnotesize Alg.~2a} & {\footnotesize $332\;(0.93$s$)$} & {\footnotesize $919\;(107$s$)$}  & {\footnotesize $1150\;(650$s$)$}
& {\footnotesize$1239\;(1973$s$)$}\\\hline
{\footnotesize Alg.~2b} & {\footnotesize $299\;(0.89$s$)$}  & {\footnotesize $812\;(81$s$)$}   & {\footnotesize $891\;(541$s$)$}
& {\footnotesize $1119\;(1370$s$)$}\\\hline
\end{tabular}
\end{center}
\vspace{-2mm}
\caption{\label{tab2}
Averaged number of iterations for Alg.~1 and Alg.~2a/b for different spatial
resolutions $d_1\!\times\!d_2\!\times\!d_3$ and the benchmark problem
over the time interval $[0,1]$. Averaged computation times are in brackets.}
\end{table}}

\begin{figure}[h!ptb]
\includegraphics[width=8.0cm]{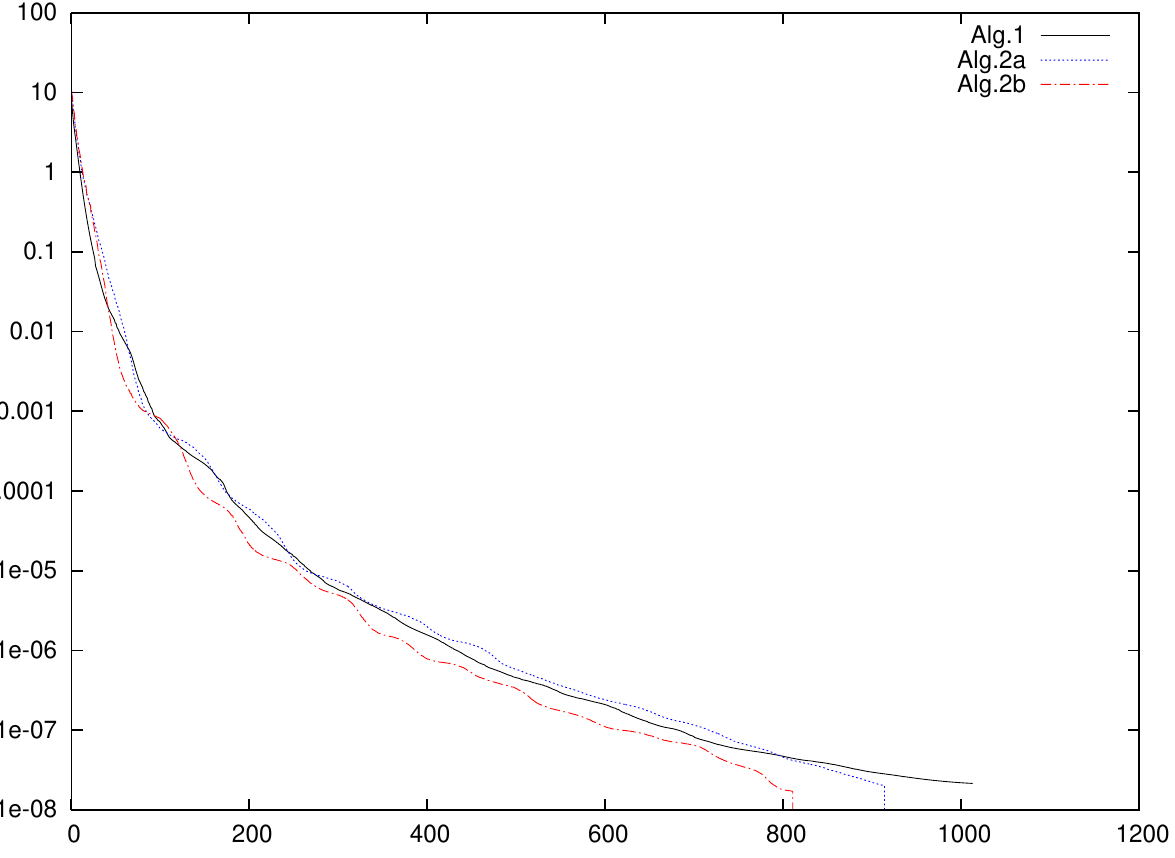} 
\includegraphics[width=8.0cm]{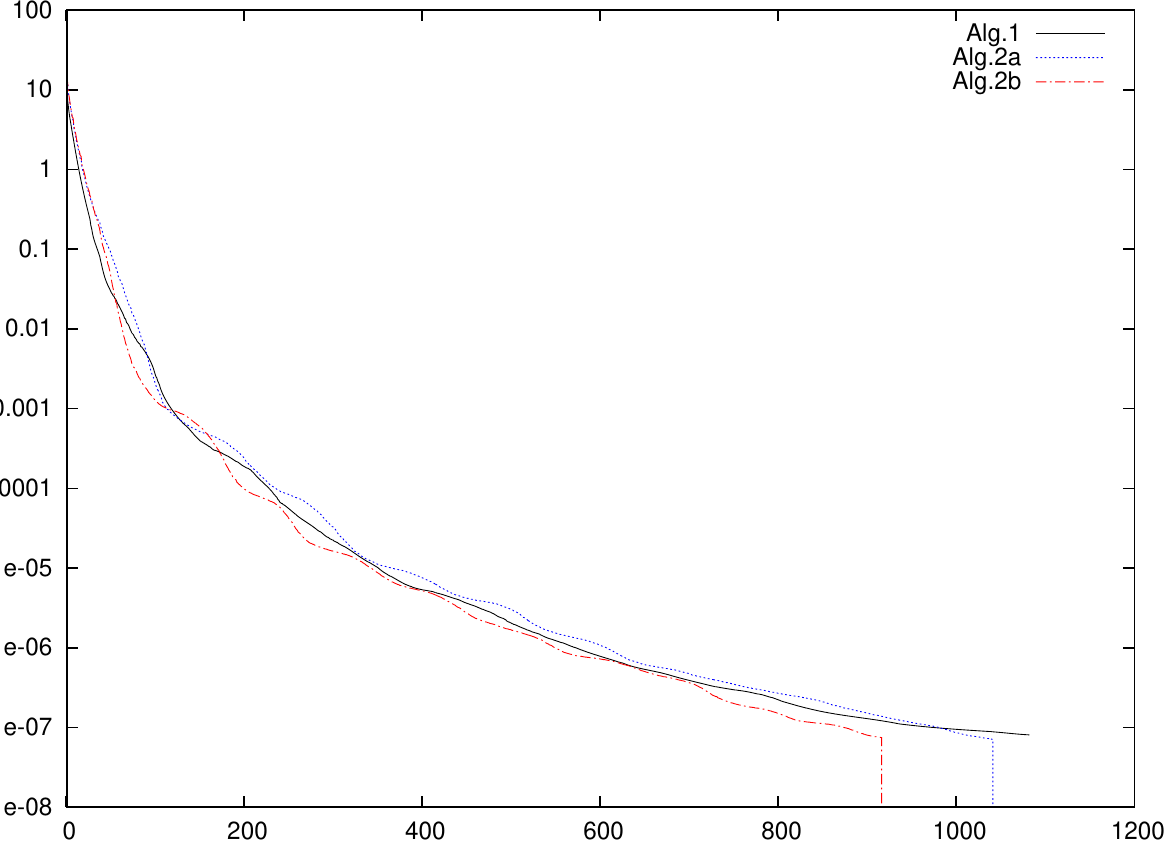} 
\includegraphics[width=8.0cm]{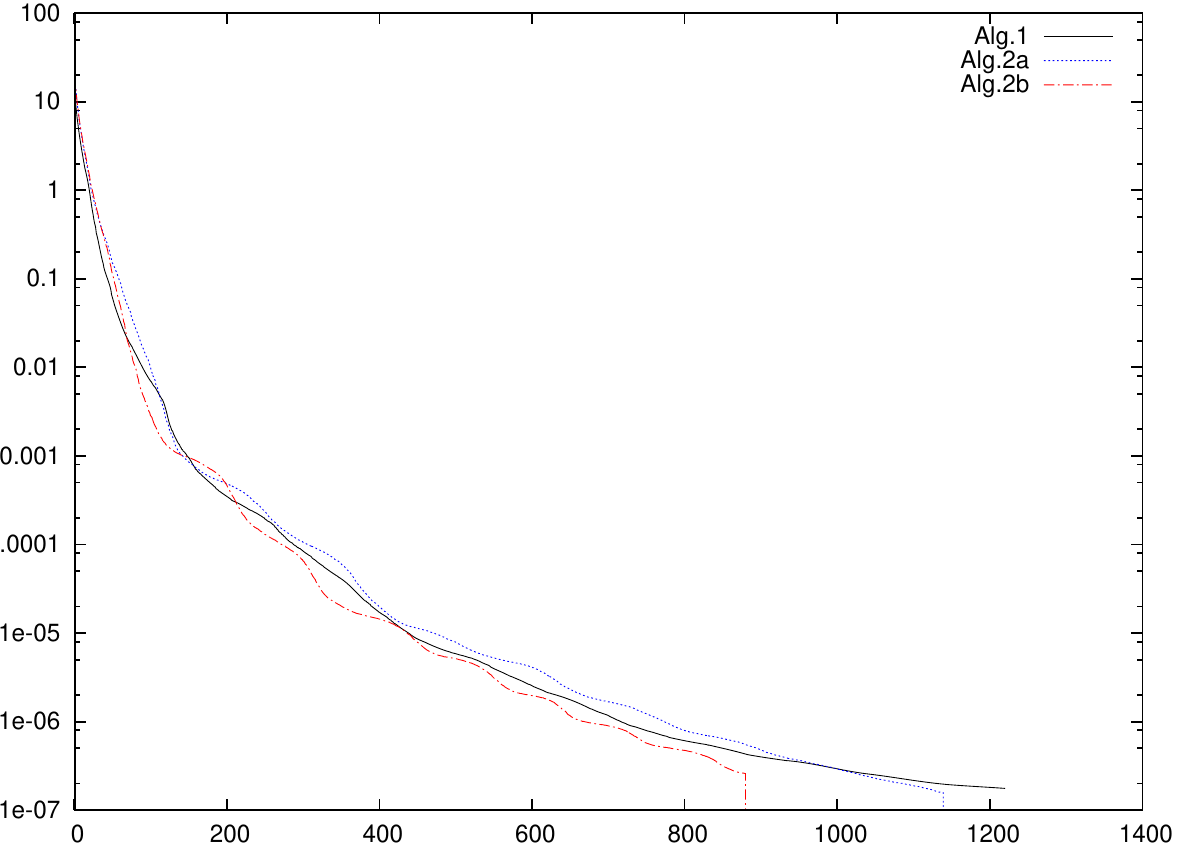}
\caption{
Convergence of Alg.~1 and the two variants of Alg.~2 for the first time step of
the benchmark problem. The values of $\cE_\ve$ are rendered on the
ordinate as a function of the L-BFGS-iterations on the abscissa.
Top left: Spatial resolution $d_1=d_2=d_3=30$.
Top right: Resolution $d_1=d_2=d_3=40$. Bottom: Resolution $d_1=d_2=d_3=50$.
The exact solution in all cases is $\cE_\ve=0$.}
\label{fig1}
\end{figure}

{\renewcommand{\arraystretch}{1.40}
\begin{table}[h!bt]
\begin{center} \tabcolsep=0.35mm
\begin{tabular}{|c|c|c|c|c|}\hline
\multicolumn{1}{|c|}{\mbox{\footnotesize Algorithm}} &
\multicolumn{1}{|c|}{\mbox{\footnotesize $10\times10\times10$}} &
\multicolumn{1}{|c|}{\mbox{\footnotesize $30\times30\times30$}} &
\multicolumn{1}{|c|}{\mbox{\footnotesize $50\times50\times50$}} &
\multicolumn{1}{|c|}{\mbox{\footnotesize $70\times70\times70$}}\\
\hline\hline
{\footnotesize Alg.~2a}\quad & {\footnotesize$8.42\cdot10^{-9}$} & {\footnotesize$4.64\cdot10^{-8}$}
& {\footnotesize$5.01\cdot10^{-7}$} & {\footnotesize$1.04\cdot10^{-6}$}\\[-2mm]
\quad{\footnotesize($\ve_0=10^{-7}$)} &   &  
&   &  \\\hline
{\footnotesize Alg.~2a}\quad & {\footnotesize$8.21\cdot10^{-14}$} & {\footnotesize$4.81\cdot10^{-12}$}
& {\footnotesize$1.22\cdot10^{-12}$} & {\footnotesize$3.05\cdot10^{-11}$}\\[-2mm]
\quad{\footnotesize( $\ve_0=10^{-9}$)} &   &  
&   &  \\\hline
\end{tabular}
\end{center}
\vspace{-2mm}
\caption{\label{tab3}
Value of $\max\limits_{t\in[0,1]}\io\Lambda(|q(x,t)|^2-1)^2\ddx$ for different
spatial resolutions, two stop values (cf. Eqn.~(\ref{lstop})), and Alg.~2a.}
\end{table}}


{\renewcommand{\arraystretch}{1.20}
\begin{table}[h!t]
\begin{center}\tabcolsep=0.45mm
\begin{tabular}{|c|r|r|r|r|}\hline
\multicolumn{1}{|c|}{\mbox{\footnotesize Resolution}} &
\multicolumn{1}{|c|}{\mbox{\footnotesize Iterations }} &
\multicolumn{1}{|c|}{\mbox{\footnotesize Iterations }} &
\multicolumn{1}{|c|}{\mbox{\footnotesize Time }} &
\multicolumn{1}{|c|}{\mbox{\footnotesize Time }}\\[-2mm]
\multicolumn{1}{|c|}{\mbox{\footnotesize }} &
\multicolumn{1}{|c|}{\mbox{\footnotesize L-BFGS}} &
\multicolumn{1}{|c|}{\mbox{\footnotesize  pc-L-BFGS}} &
\multicolumn{1}{|c|}{\mbox{\footnotesize L-BFGS}} &
\multicolumn{1}{|c|}{\mbox{\footnotesize pc-L-BFGS}}\\
\hline\hline
{\footnotesize$10\times10\times10$} & {\footnotesize$22166995$}  & {\footnotesize$41/13554468$} & {\footnotesize$383\mbox{min}\;28$s}
& {\footnotesize$234\mbox{min}\;37$s} \\\hline
{\footnotesize$20\times20\times20$} & {\footnotesize$15773252$}  & {\footnotesize$133/162642$} & {\footnotesize$2656\mbox{min}\;31$s}
& {\footnotesize$24\mbox{min}\;56$s} \\\hline
{\footnotesize$30\times30\times30$} & {\footnotesize$62300391$}  & {\footnotesize$269/229012$} & {\footnotesize$41131\mbox{min}\;17$s}
& {\footnotesize$128\mbox{min}57$s}\\\hline
\end{tabular}
\end{center}
\vspace{-2mm}
\caption{\label{tab4}
The first time step of the bending problem for the original ('L-BFGS') and the
preconditioned ('pc-L-BFGS') scheme in comparison for $\ve_0=10^{-11}$.
For the preconditioned scheme, both predictor and corrector iterations are
listed. 'Time' is the total computation time for the solution of one time step.}
\end{table}}

\begin{figure}[h!ptb]
\includegraphics[width=8.5cm]{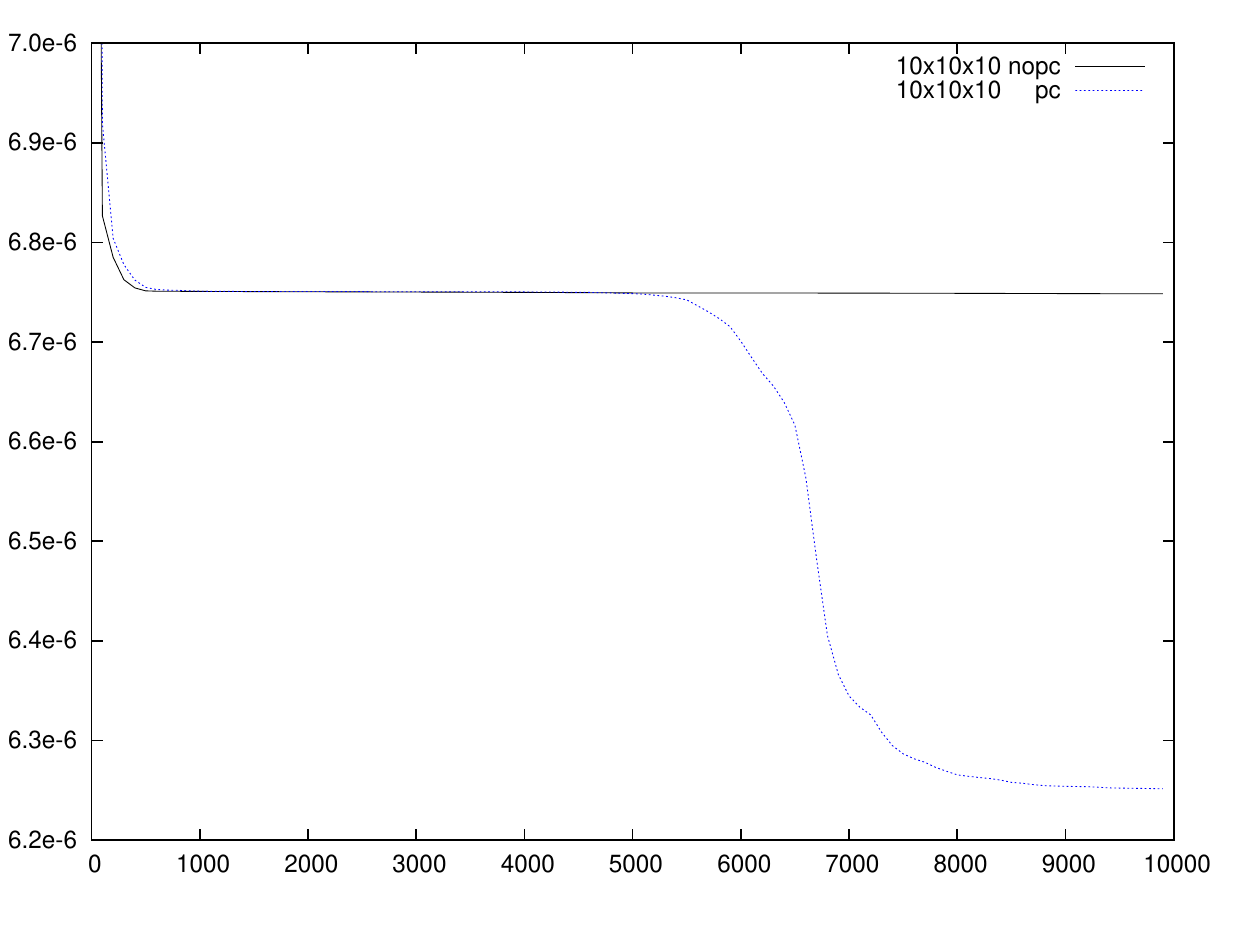}
\caption{
The progression of $\cE_\ve$ (ordinate) for the first $10000$
L-BFGS-iterations (abscissa) of the original L-BFGS method (black) and the
preconditioned L-BFGS method (blue) for $d_1\!=\!d_2\!=\!d_3\!=\!10$,
$\ve_0\!=\!10^{-11}$, and the first time step of the bending problem.
As can be seen, even the preconditioned algorithm
requires many iterations to overcome local minima of the energy.
}
\label{fig2}
\end{figure}

\vspace*{-3mm}
\begin{acknowledgements}
Part of this article was written while TB visited the Hausdorff Research
Institute for Mathematics (HIM), University of Bonn, in 2019. This visit was
supported by the HIM. TB gratefully acknowledges both this support and the
hospitality of HIM. AA gratefully acknowledges financial support from the
Italian Ministry of Education, University and Research (MIUR) under the
`Departments of Excellence' grant L.232/2016.
\end{acknowledgements}

\section*{Appendix - List of symbols}
\vspace*{-8mm}
\begin{Tabbing}
\TAB= \hspace*{40pt} \TAB= \hspace*{210pt} \TAB= \hspace*{40pt} \TAB= \kill \\
\TAB>$A\!:\!B$ \TAB>tensor product of $A$, $B$, below (\ref{Vdef})\\
\TAB>$\mathbf{u}\!\cdot\!\mathbf{v}$ \TAB>inner product of $\mathbf{u}$,
$\mathbf{v}\in\R^3$\\
\TAB>$\sym(\sigma)$ \TAB>symmetric part of a tensor $\sigma$, (\ref{Wstdef})\\
\TAB>$\skw(\sigma)$\TAB>skew-symmetric part of $\sigma$, (\ref{Wstdef})\\
\TAB>$\tr(\sigma)$ \TAB>trace of tensor $\sigma$\\
\TAB>$\sigma^t$ \TAB>transpose of $\sigma$;
$R^t\!=\!R^{-1}$ for $R\!\in\!\SO(3)$\\
\TAB>$\|\cdot\|$ \TAB>Frobenius matrix norm, (\ref{Wstdef})\\
\TAB>$|\cdot|$ \TAB>Euclidean vector norm in $\R^4$, (\ref{qnorm})\\
\TAB>$\Omega\subset\R^3$ \TAB>reference domain, undeformed solid\\
\TAB>$(x,t)$ \TAB>space and time coordinates\\
\TAB>$\phi$ \TAB>deformation vector of the solid, (\ref{split})\\
\TAB>$F\!=\!D\phi$ \TAB>deformation tensor, (\ref{split})\\
\TAB>$\Fe$ \TAB>elasticity tensor, (\ref{split})\\
\TAB>$\Fp$ \TAB>plasticity tensor, (\ref{split})\\
\TAB>$\Re$ \TAB>rotation tensor, (\ref{para2}), (\ref{para1}), (\ref{split})\\
\TAB>$\Ue$ \TAB>(right) stretching tensor, (\ref{split})\\
\TAB>$\Ke$ \TAB>(right) curvature tensor, (\ref{Kmdef})\\
\TAB>$\Id$ \TAB>identity tensor, $(\Id)_{kl}=(\delta_{kl})_{kl}$,
(\ref{Fpdef})\\
\TAB>$\alpha$ \TAB>Euler angle parameterization of $\Re$, (\ref{para2})\\
\TAB>$\gamma$ \TAB>single-slip parameterization of $\Fp$, (\ref{Fpdef})\\
\TAB>$q$ \TAB>Quaternion parameterization of $\Re$, (\ref{para1})\\
\TAB>$q_D$ \TAB>Dirichlet boundary values of $q$, (\ref{BC})\\
\TAB>$\cE$ \TAB>mechanical energy, (\ref{Emin1})\\
\TAB>$h>0$ \TAB>discrete (fixed) time step, (\ref{Emin1})\\
\TAB>$\gamma^0$ \TAB>values of $\gamma$ at old time $t$, (\ref{kapdef})\\
\TAB>$\kappa^0$ \TAB>values of $\kappa$ at old time $t$, (\ref{kapdef})\\
\TAB>$\kappa$ \TAB> dislocation density, (\ref{kapdef})\\
\TAB>$V(\kappa)$ \TAB>dislocation energy, (\ref{Vdef})\\
\TAB>$\Wst$ \TAB>stretching energy, (\ref{Wstdef})\\
\TAB>$\WWc$ \TAB>curvature energy, (\ref{Wcdef})\\
\TAB>$X$ \TAB>back stress (dual variable to $\Fp$), (\ref{LFdual}),\\
\TAB>$\xi$ \TAB>hardening (dual variable to $\kappa$), (\ref{LFdual})\\
\TAB>$\fext$ \TAB>external volume forces, (\ref{Emin1})\\
\TAB>$\Mext$ \TAB>external volume couples, (\ref{Emin1})\\
\TAB>$\sigma_Y$ \TAB>yield stress, (\ref{Emin1})\\
\TAB>$Q^*$ \TAB>dissipated energy, (\ref{Qstardef})\\
\TAB>$\mathbf{m}$ \TAB>slip vector, (\ref{Fpdef})\\
\TAB>$\mathbf{n}$ \TAB>slip normal, (\ref{Fpdef})\\
\TAB>$\rho>0$ \TAB>dislocation energy constant, (\ref{Emin1})\\
\TAB>$g_{_D}$ \TAB>Dirichlet boundary values of $\phi$, (\ref{BC})\\
\TAB>$\ve>0$ \TAB>regularization of $|\cdot|$, Remark~\ref{rem1}\\
\TAB>$\Lambda>0$ \TAB>Lagrange parameter to $|q|^2=1$, (\ref{Emin1})\\
\TAB>$\lambda$, $\mu$ \TAB> Lam{\'e} parameters, (\ref{Wstdef})\\
\TAB>$\mu_c$ \TAB>Cosserat couple modulus, (\ref{Wstdef})\\
\TAB>$L_c$ \TAB> internal length scale, (\ref{Wcdef})\\
\TAB>$\mu_2$ \TAB> parameter $\mu$ scaled by $L_c^2$, (\ref{Wcdef})\\
\TAB>$d_1,\!d_2,\!d_3$ \TAB>spatial resolution, (\ref{NC})\\
\TAB>$\eta_1\!,\!\eta_2,\!\!\eta_3$ \TAB>points on the numerical mesh,
(\ref{xyz})\\
\TAB>$N_{IJK}$ \TAB>discrete numerical weights, (\ref{NC}) \\
\TAB>$\beta(t)$ \TAB>deformation parameter, (\ref{shearBC}), (\ref{phiwarp}).
\end{Tabbing}

\section*{Compliance with Ethical Standards}
The authors declare that they have no conflict of interest.



\end{document}